\documentclass[11pt]{article}

\usepackage[preprint]{acl}

\usepackage{times}
\usepackage{latexsym}

\usepackage[T1]{fontenc}

\usepackage[utf8]{inputenc}

\usepackage{microtype}

\usepackage{inconsolata}

\usepackage{graphicx}

\usepackage{booktabs}

\title{Instructions for *ACL Proceedings}

\author{First Author \\
  Affiliation / Address line 1 \\
  Affiliation / Address line 2 \\
  Affiliation / Address line 3 \\
  \texttt{email@domain} \\}
\usepackage{amsmath,amssymb,amsthm}

\newcommand{\SCH}{\textsc{sch}}
\newcommand{\MAB}{\textsc{mab}}
\newcommand{\CAT}{\textsc{cat}}
\newcommand{\MST}{\textsc{mst}}
\newcommand{\LOFT}{\textsc{loft}}
\newcommand{\IRT}{\textsc{irt}}
\newcommand{\ILP}{\textsc{ilp}}

\newcommand{\TIF}{\ensuremath{\mathrm{TIF}}}
\newcommand{\bPhi}{\boldsymbol{\Phi}}
\newcommand{\bSigma}{\boldsymbol{\Sigma}}
\DeclareMathOperator*{\argmax}{arg\,max}
\theoremstyle{plain}\newtheorem{proposition}{Proposition}

\title{Stochastic Constrained Test Assembly\\
       for AI-Enabled Assessment Systems}

\author{Alina A. von Davier\\
  {Duolingo}{} \\
  \texttt{{avondavier@duolingo.com}{}}}

\begin{document}
\maketitle

\begin{abstract}
Test assembly, the process of constructing a complete test form from an item pool subject to
blueprint constraints, has traditionally been treated as a static optimization
problem. In AI-enabled assessment environments, however, item pools evolve
continuously as newly generated items enter with uncertain psychometric parameters,
and delivery is on demand. These conditions make test assembly a
\emph{sequential decision-making problem under uncertainty}: which form should be
deployed now, given current but incomplete knowledge of item quality, to
simultaneously maximize measurement precision, satisfy content-blueprint
constraints, maintain pool sustainability, and accelerate calibration of uncertain
new items?
This paper proposes the Stochastic Constrained Hybrid (\SCH) framework as a
principled answer to this question. \SCH{} recasts form-level assembly as a multi-armed bandit
(\MAB) problem with Fisher information as the reward, extending recent
item-level approaches in computerized adaptive testing (\CAT) to the form-level
setting. 
A simulation study comparing six test assembly methods is also presented. 
The main contribution of this paper is a framework for incorporating items with uncertain parameters into the automatic test assembly process for linear test forms.

\end{abstract}


\section{Introduction}
\label{sec:intro}

Test assembly, the process of constructing a form from an item pool subject to content
and psychometric constraints, has traditionally relied on three paradigms:
random parallel-form (\textsc{rpf}) assembly \citep{lord1964parallel},
Linearly On-the-Fly Testing (\LOFT) \citep{kingsbury1989,way1998}, and
constrained optimization via shadow testing \citep{vanderLinden2005}.
None was designed for \emph{AI-enabled assessment} environments, where
large language models continuously generate items that enter the operational
pool with uncertain psychometric parameters and no or little response data.

The challenge is direct: items cannot receive reliable difficulty or
discrimination estimates until examinees respond, yet they must be deployed
immediately to refresh the pool.
Traditional solutions---a pilot phase or separate field study---are expensive,
create security risks, and artificially separate calibration from
administration.

We propose the \textbf{Stochastic Constrained Hybrid (\SCH) framework}.
Drawing on \citet{sharpnack2024banditcat}, the core insight is that form
assembly can be recast as a \emph{multi-armed bandit} (\MAB) problem.

A multi-armed bandit is a sequential decision-making framework in which an agent
repeatedly chooses among a set of actions (``arms''), each yielding a stochastic
reward, and must balance \emph{exploitation} (choosing the action currently
believed to be best) against \emph{exploration} (trying less-certain actions to
gather information). The name comes from the image of a row of slot machines
(``one-armed bandits''): pulling each lever gives a random payoff, and the
agent must learn which lever pays best while minimizing losses during learning.
In this paper, each arm is one complete test form---a specific selection of
$n=40$ items satisfying all blueprint constraints---and the reward is the
expected measurement precision of that form.
Thompson sampling~\cite{thompson1933} is a principled algorithm for solving
\MAB{} problems: for each bandit arm it maintains a probability distribution over
the arm's true reward, draws a random sample from each distribution, and
selects the arm whose sample is highest. Arms whose rewards are highly
uncertain receive occasionally large random draws and are therefore explored;
arms with well-understood rewards are selected only when they genuinely
dominate. In the test-assembly context, items with uncertain psychometric
parameters (new, insufficiently calibrated items) naturally receive elevated
Thompson samples (higher probability to be selected) and are therefore routed to examinees for calibration,
without requiring a separate pilot phase.

The \SCH{} framework extends this idea from selecting \emph{one item} per step
in a computerized adaptive test (\CAT) to selecting a complete \emph{$n$-item
form} subject to a test blueprint specifying content-category quotas and
difficulty-band constraints. This is an example of the computational
psychometrics approach~\cite{vonDavier2021computational}, combining traditional
psychometric methods with machine learning algorithms.

The extension of the \MAB{} method from single-item in a computerized adaptive test (\CAT{}, \citep{vanderLinden1998}) as it was developed by \citet{sharpnack2024banditcat} to
$n$-item form assembly is non-trivial for three reasons:
(i)~the \MAB{} arm space
$|\mathcal{C}|{=}\prod_\sigma\binom{|\mathcal{I}_\sigma|}{n_\sigma}\approx10^{68}$
is intractable by enumeration, considering the size of the full form-selection arm space being the number of distinct blueprint-feasible test forms that could be assembled from the pool. In this paper we show how to reduce this number to a manageable number of arms per
blueprint stratum cell, for a specific, but realistic case, making Thompson sampling computationally feasible. 
(ii)~blueprint constraints couple item choices across content categories and
difficulty bands; and
(iii)~post-selection constraint repair creates a learning inconsistency in
the bandit update.
\SCH{} resolves all three.

\paragraph{Contributions.}
(1)~A formal extension of bandit item selection to form-level assembly;
(2)~a stratified cell decomposition (with a pragmatic choice) reducing $10^{68}$ arms to ${\leq}70$
per cell while preserving the information-maximization objective;
(3)~a full delta-method variance approximation over all three \IRT{}
parameters, correcting an $88\%$ underestimate from difficulty-only methods;
(4)~a new information-weighted \LOFT{} baseline revealing a previously
undescribed intermediate operating regime; and
(5)~evidence that pool concentration is a tunable trade-off governed by the
exposure-penalty weight.

\section{Background and Related Work}
\label{sec:background}

\paragraph{IRT and test information.}
Under the 3PL model, the probability that examinee $j$ with ability $\theta_j$
answers item $i$ correctly is
$P_i(\theta){=}c_i{+}(1{-}c_i)/(1{+}e^{-a_i(\theta-b_i)})$,
where $\bPhi_i{=}(a_i,b_i,c_i)$ are discrimination, difficulty, and
pseudo-guessing.
Fisher information is
$I_i(\theta;\bPhi_i){=}a_i^2[P_i{-}c_i]^2[1{-}P_i]/[(1{-}c_i)^2P_i]$.
The \emph{expected test information} of form $F$ is
\begin{equation}
  \mathcal{R}(F)
    = \int\!\sum_{i\in F}I_i(\theta;\bPhi_i)\,\pi(\theta)\,d\theta,
  \label{eq:etif}
\end{equation}
where $\pi(\theta){\sim}\mathcal{N}(0,1)$ (approximated by 200-point
Gaussian quadrature); higher $\mathcal{R}(F)$ means more precise ability
estimates.

\paragraph{Blueprint and stratum cells.}
A blueprint $\mathcal{B}$ specifies form length $n$, per-category counts
$\{n_k\}$ ($\sum_k n_k{=}n$), difficulty-band bounds, a target mean-difficulty
interval $[\underline{b},\bar{b}]$, and item exposure ceiling $\rho_{\max}$.
The $S{=}K{\times}L$ intersections of $K$ content categories and $L$ difficulty
bands are \emph{stratum cells} $\sigma$, each with quota $n_\sigma$ and item
pool $\mathcal{I}_\sigma$.

\paragraph{Related work.}
\citet{vanderLinden1998} formulated Bayesian \CAT{} item selection.
Shadow testing \citep{vanderLinden2005} added blueprint feasibility via an
\ILP{} at each selection step.
\citet{sharpnack2024banditcat} reframed \CAT{} selection as a bandit problem
(BanditCAT), enabling simultaneous precision and calibration.
\citet{sharpnack2025s2a3} added stochastic Sympson--Hetter control (the SAC
framework), validated on the Duolingo English Test.
This paper extends both from single-item \CAT{} to form-level assembly.

\section{Assembly Methods Compared}
\label{sec:methods}

Stratified Random Parallel Form ,\textbf{SR} \citep{lord1964parallel} draws $n_\sigma$ items uniformly from
each cell item pool  $\mathcal{I}_\sigma$.

\LOFT with exposure control, \textbf{LOFT-SH} \citep{kingsbury1989,way1998} draws uniformly subject to a
Sympson--Hetter (S-H) constraint \citep{sympsonHetter1985} keeping $\rho_i{\leq}\rho_{\max}$
via per-item eligibility probabilities $\alpha_i$ updated every $W$ cycles.
\LOFT that considers the expected information function \textbf{LOFT-IW} (new baseline) samples within cells proportional to expected
information $\mathcal{R}(\{i\}){=}\int I_i(\theta;\hat{\bPhi}_i)\pi(\theta)d\theta$,
subject to the same S-H constraint.
Shadow test applied to a linear test form \textbf{ST-Linear} \citep{vanderLinden2005} solves one binary \ILP{} per cycle:
$\max_{\mathbf{x}\in\{0,1\}^N}\mathbf{v}^\top\mathbf{x}$ subject to blueprint
constraints, where $v_i{=}\mathcal{R}(\{i\}){-}\lambda_e\rho_i$ ($\lambda_e{=}0.5$).
\textbf{SCH-MAB} and \textbf{SCH-EXP} are the new methods that use \MAB and they are described in Section~\ref{sec:sch}.

\section{The SCH Framework}
\label{sec:sch}

\subsection{Utility Score and Arm Reduction}

Each item receives a multi-objective utility score:

\begin{align}
  u_i(t) = \underbrace{\lambda_1 \mathbb{E}_{\pi}[I_i(\theta;\hat{\bPhi}_i)]}_{\substack{\text{info}\\
             \text{(precision)}}}
           +\underbrace{\lambda_2\,\mathrm{tr}(\hat{\bSigma}_i)}_{\substack{\text{uncertainty}\\
             \text{(calibration)}}}    
             \notag\\
           +\underbrace{\lambda_3 \exp\!\bigl(-\alpha(t-t_i^{\mathrm{add}})\bigr)}_{\substack{
             \text{refresh}\\\text{(novelty)}}}
           -\underbrace{\lambda_4\rho_i(t)}_{\substack{\text{exposure}\\
             \text{(sustainability)}}},
  \label{eq:utility}
\end{align}

\noindent where $\hat{\bSigma}_i$ is item $i$'s posterior parameter covariance
(calibration uncertainty), $t_i^{\mathrm{add}}$ is the administration cycle in which item $i$
        entered the pool; $\alpha > 0$ is the novelty decay rate (novelty),
and $\rho_i(t)$ its exposure rate (pool health).

\begin{proposition}[Reward decomposition]
\label{prop:decomp}
Under an \IRT{} model, the expected test information of a form decomposes
additively over stratum cells:
$\mathcal{R}(F)=\sum_\sigma\mathcal{R}(F_\sigma)$ where
$F_\sigma{=}F\cap\mathcal{I}_\sigma$.
\end{proposition}
\noindent\textit{Proof.} $\TIF(\theta;F)=\sum_{i\in F}I_i(\theta)=\sum_\sigma\sum_{i\in F_\sigma}I_i(\theta)
=\sum_\sigma\TIF(\theta;F_\sigma)$. Integrating against $\pi(\theta)$ and applying
linearity of integration gives $\mathcal{R}(F)=\sum_\sigma\mathcal{R}(F_\sigma)$. Hence, linearity of integration over $\pi(\theta)$ proves it.\quad\qed

In order to make the problem computationally manageable, test form assembly therefore decomposes into $S$ independent per-cell problems.
The top-$s{=}8$ items (pragmatically chosen here) in each cell by $u_i$ form candidate set
$\mathcal{Q}_\sigma$; the arm set
$\mathcal{A}_\sigma{=}\binom{\mathcal{Q}_\sigma}{n_\sigma}$ has at most
$\binom{8}{4}{=}70$ arms.
Setting $s{=}8$ satisfies $s{\geq}2n_\sigma$ for all $n_\sigma{\in}\{2,3,4\}$
in the $n{=}40$ blueprint ($K{=}5$, $L{=}3$; item allocation $[10,8,8,8,6]$).

\subsection{Delta-Method Variance for Exploration}

For Thompson sampling to explore uncertain items, reward variance must
reflect parameter uncertainty:

\begin{align}
  \widehat{\mathrm{Var}}\bigl(\mathcal{R}(\mathfrak{a})\bigr)
    &\approx
    \sum_{i\in\mathfrak{a}}\!\Bigl[
      \Bigl(\tfrac{\partial\mathcal{R}}{\partial a_i}\Bigr)^{\!2}\!\mathrm{SE}(\hat{a}_i)^2
     +\Bigl(\tfrac{\partial\mathcal{R}}{\partial b_i}\Bigr)^{\!2}\!\mathrm{SE}(\hat{b}_i)^2
     \notag\\
    &\quad\;\;
     +\Bigl(\tfrac{\partial\mathcal{R}}{\partial c_i}\Bigr)^{\!2}\!\mathrm{SE}(\hat{c}_i)^2
    \Bigr],
  \label{eq:deltamethod}
\end{align}

where $\mathfrak{a}$ is a candidate arm (an $n_\sigma$-item subset of
$\mathcal{Q}_\sigma$) and $\mathrm{SE}(\hat{a}_i),\mathrm{SE}(\hat{b}_i),
\mathrm{SE}(\hat{c}_i)$ are standard errors of the item parameter estimates.
A $b$-only approximation captures only $11.6\%$ of full variance for
jump-start items ($\mathrm{SE}(\hat{b}){\approx}0.4$) because
$|\partial\mathcal{R}/\partial a_i|{\gg}|\partial\mathcal{R}/\partial b_i|$:
discrimination uncertainty dominates.
Without the full variance, Thompson sampling collapses to pure exploitation and will avoid the new items (exploration). 

\subsection{SCH-MAB and SCH-EXP}

The \MAB extension to test form assembly, \textbf{SCH-MAB} maintains a Gaussian approximation for each arm
$\mathfrak{a}$: mean $\mu_{\sigma,\mathfrak{a}}{=}\mathcal{R}(\mathfrak{a})|_{\bPhi_i=\hat{\bPhi}_i}$
and variance $\sigma^2_{\sigma,\mathfrak{a}}$ from Eq.~\eqref{eq:deltamethod}.
It draws $\tilde{\mathcal{R}}_{\sigma,\mathfrak{a}}{\sim}
\mathcal{N}(\mu_{\sigma,\mathfrak{a}},\sigma^2_{\sigma,\mathfrak{a}})$
and selects $\mathfrak{a}_\sigma^*{=}\argmax_\mathfrak{a}
\tilde{\mathcal{R}}_{\sigma,\mathfrak{a}}$.
Arms with large $\sigma^2$ (uncertain items) receive extreme draws and are
explored; well-calibrated arms are selected only when their mean dominates.
Exposure is controlled through $\lambda_4$ in Eq.~\eqref{eq:utility}.

\MAB with exposure control  is the extension of \textbf{SCH-EXP} \citep{sharpnack2025s2a3} from an item level \CAT to a linear test that adds Sympson-Hetter (S-H) arm weights:
$\mathfrak{a}_\sigma^*{=}\argmax_\mathfrak{a}\;
w_\mathfrak{a}{\cdot}\max(\tilde{\mathcal{R}}_{\sigma,\mathfrak{a}},0)^\beta$,
where $w_\mathfrak{a}{=}\min_{i\in\mathfrak{a}}\alpha_i$ and $\beta{=}0.5$.
This provides a harder exposure ceiling than $\lambda_4$ alone.

\textbf{Feasibility filter.}
Filtering arms to satisfy the mean-difficulty constraint \emph{before}
Thompson sampling eliminates a systematic algorithm learning bias: without it the bandit
updates on a repaired arm rather than the one it chose.
This improved the constraint satisfaction, CS, from $0.825$ to $0.963$ for SCH-MAB.

\section{Simulation Study}
\label{sec:simulation}

\paragraph{Design.}
$N{=}1{,}000$ items; $n{=}40$; $K{=}5$ content categories; $L{=}3$ difficulty bands;
$T{=}500$ cycles; $J{=}500$ examinees/cycle; $R{=}30$ replications;
$\rho_{\max}{=}0.30$.
Item parameters: $a_i{\sim}\mathrm{LogNormal}(0,0.5)$,
$b_i{\sim}\mathcal{N}(0,1)$, $c_i{\sim}\mathrm{Beta}(4,16)$.
New items enter with jump-start calibration
($\mathrm{SE}(\hat{b}_i){\approx}0.30$--$0.50$).
SCH weights: $(\lambda_1,\lambda_2,\lambda_3,\lambda_4){=}(0.40,0.25,0.15,0.20)$
for SCH-MAB; $(0.55,0.25,0.15,0.05)$ for SCH-EXP;
$s{=}8$ for both; Python~3.12, single CPU.

Five scenarios (Table~\ref{tab:scenarios}) range from static (S1) to
inflow plus drift (S5).
\emph{Inflow}: new jump-start items added each cycle;
\emph{drift}: $b_i{+}{=}\varepsilon_t$, $\varepsilon_t{\sim}\mathcal{N}(0,0.05^2)$.

\begin{table}[t]
\centering
\caption{Simulation scenarios.}
\label{tab:scenarios}
\setlength{\tabcolsep}{3.5pt}\small
\begin{tabular}{@{}lllr@{}}
\toprule
\textbf{ID} & \textbf{Description}  & \textbf{Inflow}   & \textbf{Drift}  \\
\midrule
S1 & Static pool          & None              & None            \\
S2 & Moderate inflow      & 5\%/cycle         & None            \\
S3 & High inflow          & 20\%/cycle        & None            \\
S4 & Drift only           & None              & $\sigma{=}0.05$ \\
S5 & Inflow $+$ drift     & 10\%/cycle        & $\sigma{=}0.05$ \\
\bottomrule
\end{tabular}
\end{table}

\paragraph{Metrics.}
Form comparability (FC$\downarrow$: SD of $\mathcal{R}(F_t)$),
constraint satisfaction (CS$\uparrow$),
average test information (ATI$\uparrow$: mean $\mathcal{R}(F_t)$),
pool sustainability via exposure Gini (IPS$\downarrow$),
new-item calibration (NIC$\uparrow$: mean $\mathrm{SE}(\hat{b})$ reduction),
subgroup fairness via differential test functioning (SF$\downarrow$), and
computational cost (CC$\downarrow$, ms/cycle).

\section{Results}
\label{sec:results}

Table~\ref{tab:main} reports cross-scenario means and
Figure~\ref{fig:tradeoff} visualizes the ATI--IPS trade-off.
Three distinct operating regimes emerge.

\paragraph{Regime 1 (SR, LOFT-SH).}
ATI${=}0.228$, IPS${=}0.177$.
High pool health, modest precision.
SR and LOFT-SH are statistically equivalent on measurement quality;
LOFT-SH's advantage is the provable S-H exposure ceiling.
NIC${=}0.318$ arises from new items entering the pool by chance.

\paragraph{Regime 2 (LOFT-IW).}
ATI${=}0.350$, IPS${=}0.422$.
A previously undescribed intermediate regime accessible with a minimal
change to any LOFT-SH system: replace uniform within-cell sampling
with information-proportional sampling.

\paragraph{Regime 3 (ST-Linear, SCH).}
ATI${\geq}0.64$, IPS${\geq}0.92$.
High precision, substantial pool concentration.
SCH-MAB matches ST-Linear's precision (ATI\,$0.656$ vs.\ $0.642$) while
providing better constraint satisfaction under drift (CS\,$0.963$ vs.\
$0.880$--$0.890$) and active new-item calibration (NIC\,$0.072$ vs.\ $0.039$).

\paragraph{Pool concentration.}
IPS${\approx}0.93$ means ${\approx}120$ items receive all usage,
leaving ${\approx}880$ essentially unused.
Sensitivity analysis shows increasing $\lambda_4$ from $0.15$ to $0.40$
reduces IPS by $33\%$ ($0.880\to0.586$) at an ATI cost of $27\%$
($0.567\to0.411$)---concentration is a tunable information--security
trade-off, not an algorithmic inevitability.
The $\beta$ temperature in SCH-EXP has negligible effect on IPS
($\beta{\in}\{0.2,0.5,1.0\}$: range $0.812$--$0.831$), confirming
that $\lambda_4$, not selection stochasticity, governs concentration.

\begin{table*}[t]
\centering
\caption{Cross-scenario means ($R{=}30$, $T{=}500$).
  Bold = best per metric.
  $^{\dagger}$NIC: inflow scenarios only.
  $^{\ddagger}$Cached production cost ${\approx}4$\,ms.}
\label{tab:main}
\setlength{\tabcolsep}{4.5pt}\small
\begin{tabular}{@{}lrrrrrrr@{}}
\toprule
\textbf{Method}
  & \textbf{FC\,$\downarrow$} & \textbf{CS\,$\uparrow$}
  & \textbf{ATI\,$\uparrow$} & \textbf{IPS\,$\downarrow$}
  & \textbf{NIC}$^{\dagger}$\textbf{\,$\uparrow$}
  & \textbf{SF\,$\downarrow$}
  & \textbf{CC\,(ms)\,$\downarrow$} \\
\midrule
SR         & 0.027 & 0.986 & 0.228 & \textbf{0.176} & 0.318 & \textbf{0.007} & \textbf{1.4} \\
ST-Linear  & \textbf{0.016} & 0.953 & 0.642 & 0.921 & 0.039 & 0.022 & 27.9 \\
LOFT-SH    & 0.027 & \textbf{0.987} & 0.228 & 0.178 & 0.318 & \textbf{0.007} & 1.7 \\
LOFT-IW$^{\ddagger}$
           & 0.032 & \textbf{0.987} & 0.350 & 0.422 & 0.177 & 0.012 & ${\approx}4$ \\
SCH-MAB    & \textbf{0.014} & 0.963 & \textbf{0.656} & 0.930 & \textbf{0.072} & 0.022 & 80.6 \\
SCH-EXP    & 0.062 & 0.980 & 0.643 & 0.926 & 0.053 & 0.022 & 100.5 \\
\bottomrule
\end{tabular}
\end{table*}

\begin{figure}[t]
  \centering
  \includegraphics[width=\columnwidth]{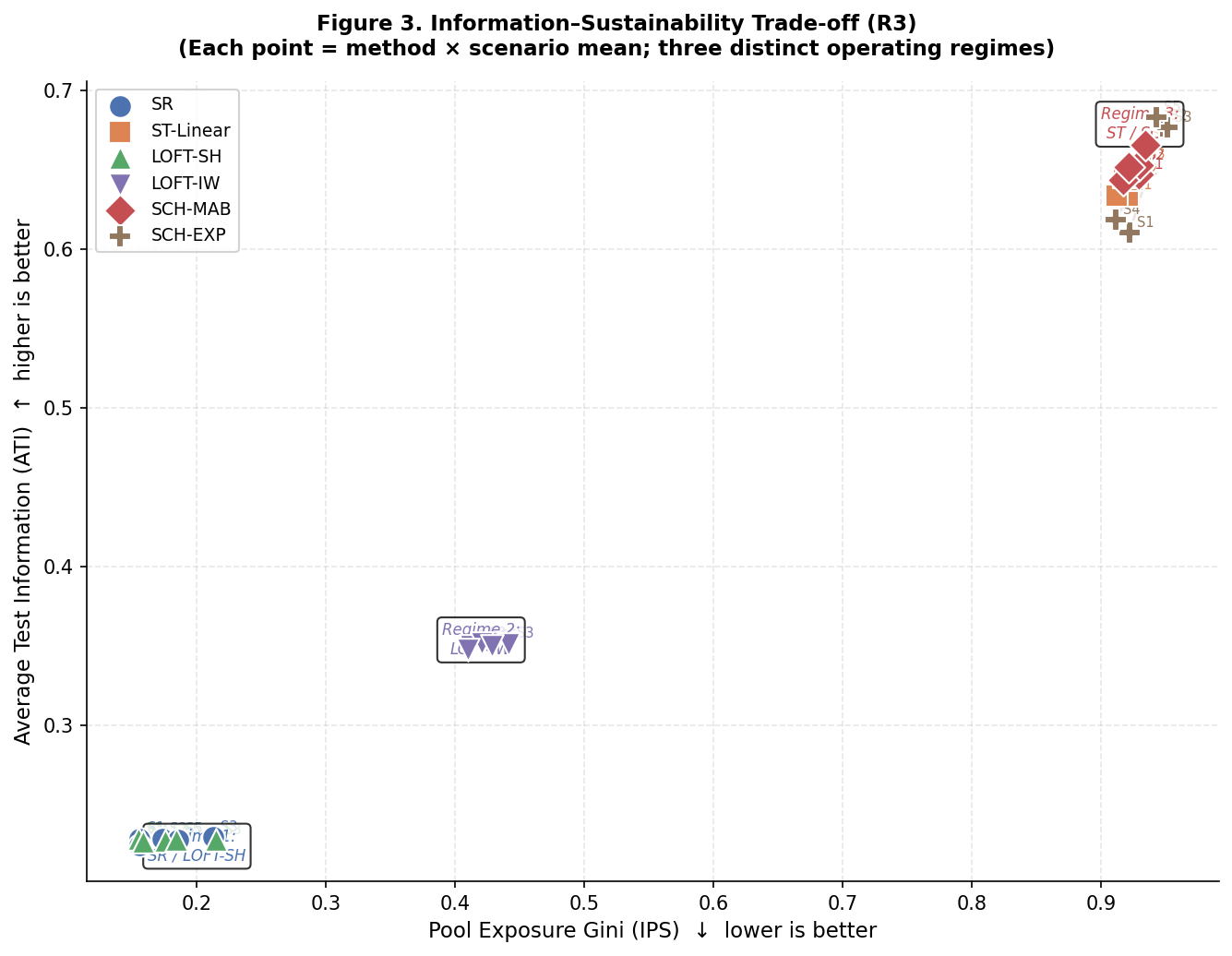}
  \caption{ATI vs.\ pool exposure Gini (IPS) across six methods and five
    scenarios. Three operating regimes are clearly separated.
    LOFT-IW (Regime 2) is a new intermediate regime.}
  \label{fig:tradeoff}
\end{figure}

\section{Discussion and Conclusion}
\label{sec:conclusion}

This paper introduced the Stochastic Constrained Hybrid (\SCH) framework for
test form assembly in AI-enabled assessment systems. The motivating question was
how to simultaneously maximize measurement precision, satisfy blueprint constraints,
maintain pool sustainability, and accelerate calibration of uncertain new
items; four objectives that no existing method addresses jointly. \SCH{} answers
this question by recasting form-level assembly as a multi-armed bandit problem,
exploiting the additivity of the \IRT{} test information function to decompose an
intractable arm problem into small, independent per-cell bandits with
a manageable number of arms each.

Three methodological advances distinguish this framework from earlier item-level
bandit approaches: a full delta-method variance approximation (correcting an $88\%$
underestimate from difficulty-only approximations); a feasibility filter that
ensures the bandit learns from exactly what it assembles; and a theoretical
demonstration---confirmed via simulations---that pool concentration is a tunable
ATI--security trade-off governed by the exposure-penalty weight $\lambda_4$.

A simulation study comparing six methods reveals three distinct operating regimes.
The corrected intermediate LOFT-IW regime (ATI\,$=0.350$, IPS\,$=0.422$) is a new finding
with immediate practical implications: it is accessible with minimal modification
of any existing \LOFT{} implementation and may serve as the appropriate choice for
programs that need more than stratified random sampling but cannot tolerate the
pool concentration of full information maximization. The two \SCH{} variants
occupy the information-maximizing Regime 3 alongside shadow testing, but offer
more stable constraint satisfaction under item-parameter drift and the ability
to calibrate uncertain new items without a separate pilot phase.

\paragraph{Limitations and future work.}
Results are simulation-based with very specific default values chosen for pragmatic reasons; real-pool validation is needed.
IPS${\approx}0.93$ limits SCH in security-sensitive programs without
tuning $\lambda_4$.
Natural extensions include \MST{} module routing
\citep{yan2014multistage,yan2024multistage}, where modules become bandit
arms ($K_g{\leq}5$ per stage), and adaptive $\lambda$ scheduling.

\SCH{} addresses a genuinely novel problem: assembling blueprint-compliant
forms that simultaneously maximize precision, sustain pool health, and
calibrate uncertain new items without a pilot phase.

\section*{Acknowledgments and Declaration on the Use of Generative AI}

This research study was an experiment in using Generative AI as a scientific collaborator, inspired by Google's Co-Scientist as described in \cite{gottweis2025aicoscientist} and other papers such as  \cite{shao2026sciscigpt}. During the preparation of this work, including the exploration of ideas, simulations, drafting, and reference reformatting, I utilized Claude, Sonnet 4.6 (Enterprise). Claude operated with access to Python 3.12 on a single CPU in a Linux (Ubuntu 24) environment. Additionally, DeepSeek–V3 was employed as a reviewer for both the paper and the code. While the primary ideas presented are my own, I interacted extensively with these technologies over several months to iterate, review, and edit all content. I take full responsibility for the final publication. 

I would like to thank Steve Nydick (Duolingo) and J.R. Lockwood (Independent Consultant) for their detailed reviews of the paper. 
Finally, I am grateful for the opportunity to have discussed several research papers on random parallel test forms at the 2026 National Council on Measurement in Education (NCME) in Los Angeles. In particular, the work presented by David Foster and Sergio Araneda (Caveon) provided the initial inspiration for this paper.

\bibliography{references_final}

\end{document}